\documentclass[journal,twoside,web]{ieeecolor}
\usepackage{kjkslim}
\usepackage{comment}
\usepackage{array}

\definecolor{subsectioncolor}{rgb}{0,0.35,0.2}

\title{Indirect data-driven predictive control and the state-space predictor}

\author{Levi D. Reyes Premer, Arash J. Khabbazi, and Kevin J. Kircher, \IEEEmembership{Senior Member, IEEE}
\thanks{The authors are with the School of Mechanical Engineering, Purdue University, 585 Purdue Mall, West Lafayette, Indiana, USA. \{{\tt \small lreyespr,arashjkh,kkirche}\}{\tt \small@purdue.edu}}
}

\begin{document}

\maketitle
\thispagestyle{empty}
\pagestyle{empty}

\begin{abstract}
We define trajectory predictive control (TPC) as a class of indirect data-driven predictive control (DDPC) methods that represent future outputs as linear in past inputs/outputs and future inputs. TPC unifies many DDPC variants with different predictor structures. We introduce a predictor with a state-space representation and show that with it, TPC inherits the mature theory of linear model predictive control. In numerical experiments, the state-space predictor outperforms existing predictors, especially for small training datasets. 
\end{abstract}

\begin{IEEEkeywords}
Data-driven predictive control, identification for control, model predictive control, system identification
\end{IEEEkeywords}

\section{INTRODUCTION}

\IEEEPARstart{M}{odel} predictive control (MPC) has been applied for decades in several industries \cite{morari1999model}, but can be costly to implement due in part to its reliance on a model that predicts how a system will respond to planned control inputs. Data-driven predictive control (DDPC) emerged as an alternative to MPC around 1999 \cite{favoreel1999spc} and has recently seen a resurgence of research activity \cite{berberich2025overview}. Broadly speaking, where MPC often uses physics-based models developed by domain experts, DDPC aims to optimize closed-loop performance given only a set of past input/output data. As gathering information-rich input/output data can be costly or risky, especially in safety-critical applications, DDPC should ideally work with small training datasets gathered in closed loop.

DDPC methods can be categorized as indirect or direct \cite{krishnan2021direct, dorfler2023data}. Indirect DDPC identifies a system model then embeds it in control optimization. Direct DDPC methods, such as data-enabled predictive control (DeePC) \cite{coulson2019data}, embed the training data directly in control optimization, bypassing system identification. DeePC is based on Willems' fundamental lemma of behavioral system theory \cite{willems2007behavioral, markovsky2021behavioral}, which represents future input/output trajectories as linear combinations of input/output trajectories from the training data. DeePC performs well for deterministic linear systems but can break down under uncertainty or nonlinearity. This observation has motivated researchers to modify DeePC with a variety of relaxations and regularizations aimed at improving robustness \cite{coulson2021distributionally, dorfler2022bridging, breschi2023data}.

DeePC modifications have progressively blurred the boundary between indirect and direct DDPC. In \cite{dorfler2022bridging}, for example, D{\"o}rfler et al. show that two direct DDPC variants are convex relaxations of indirect DDPC, with regularization corresponding to implicit system identification. In \cite{breschi2023data}, Breschi et al. show that adding a least-norm constraint effectively makes DeePC an indirect method. In \cite{chiuso2025harnessing}, Chiuso et al. establish a DDPC separation principle for linear, time-invariant (LTI) systems with quadratic costs and no constraints. In this setting, separating identification from the rest of DDPC design does not reduce DDPC's closed-loop performance. Conducting system identification separately also has the advantage of providing an explicit model with quantified uncertainty, enabling principled analysis, simulation, and tuning of controllers \cite{dorfler2023data}.

This paper collects a family of indirect DDPC methods under a unifying framework that we refer to as trajectory predictive control (TPC). While established system identification methods often focus on learning an LTI state-space model \cite{van1994n4sid, verhaegen1992subspace}, TPC methods represent the full output trajectory over the prediction horizon as a linear function of the recent input/output history and the planned input trajectory. As the recent input/output history is perfectly observed, TPC enables output-feedback control without state estimation. This paper shows how TPC encompasses a variety of indirect DDPC methods and how TPC relates to direct DDPC via DeePC.

This paper then introduces a trajectory predictor that, to the authors' knowledge, has not been studied in the DDPC literature. Unlike existing trajectory predictors, the predictor studied here corresponds to an LTI state-space model with the recent input/output history as the state. This correspondence establishes TPC with the `state-space predictor' as a special case of conventional MPC with an LTI state-space model. Conceptually, this result brings DDPC full circle, returning it to the well-studied framework of MPC. Pragmatically, this result equips DDPC with a mature body of LTI MPC theory and methods \cite{rawlings2020model}. 

In numerical experiments, TPC with the state-space predictor performs as well as TPC with any existing predictor, approaching the limit of oracle linear-quadratic-Gaussian (LQG) control with perfect knowledge of the true system model. For small training datasets, the state-space predictor outperforms other trajectory predictors because it has fewer parameters.

The rest of this paper is organized as follows. \S\ref{notation} collects notation and mathematical preliminaries. \S\ref{tpcSec} introduces the general TPC framework. \S\ref{predictorSurvey} relates TPC to existing DDPC variants. \S\ref{mainResult} introduces the state-space predictor, establishes its correspondence to an LTI state-space system, and compares it to other trajectory predictors. \S\ref{experiments} presents numerical experiments. \S\ref{futureWork} discusses possible directions for future work.

\subsection{Notation and preliminaries}
\label{notation}

The sets of real scalars, real $n$-dimensional column vectors, and real $m \times n$ matrices are $\R$, $\R^n$, and $\R^{m \times n}$, respectively. An identity matrix is $I_n \in \R^{n \times n}$ and $0_{m,n} \in \R^{m \times n}$ is a matrix of zeros. The pseudoinverse of $A \in \R^{m \times n}$ with rank$(A) = m \leq n$ is $A^\dagger = A^\top (A A^\top)^{-1} \in \R^{n \times m}$. For $A \in \R^{m \times n}$ with rank$(A) = m \leq n$ and $b \in \R^m$, $\hat x = A^\dagger b$ is the unique minimizer of $\norm{x}_2$ subject to $A x = b$, where $\norm{\cdot}_2$ is the Euclidean norm. For the linear model $Y = \Theta X + E$ with targets $Y \in \R^{m \times n}$, parameters $\Theta \in \R^{m \times p}$, features $X \in \R^{p \times n}$, errors $E \in \R^{m \times n}$, and rank($X) = p \leq n$, $\hat \Theta = Y X^\dagger$ is the unique minimizer of the mean squared error 
\[
\norm{Y - \Theta X}_\text{Fro}^2 / n = \text{trace}( (Y - \Theta X)^\top (Y - \Theta X) ) / n ,
\]
where $\norm{\cdot}_\text{Fro}$ is the Frobenius norm. The stacked column vector with $x \in \R^n$ above $y \in \R^m$ is $(x,y) = \bmat x \\ y \emat \in \R^{n + m}$.

\section{TRAJECTORY PREDICTIVE CONTROL}
\label{tpcSec}

We consider a system with discrete time index $t$. At each $t$, the controller sends an input $u(t) \in \R^{n_u}$ to the system, the system evolves to a new state (through unknown dynamics, possibly influenced by unmeasured disturbances), and the controller receives a (possibly noisy) output $y(t) \in \R^{n_y}$ from the system. We define TPC, summarized in Alg. \ref{tpcAlg}, as a family of output-feedback indirect DDPC algorithms that choose the input $u(t)$ based on the last $m$ inputs and outputs,
\bneq
z_p(t) = \bmat
z(t-m) \\
\vdots \\
z(t-1) \\
\emat , \text{ where } z(t) = \bmat u(t) \\ y(t) \\ \emat \in \R^{n_z} \label{z}
\eneq
and $n_z = n_u + n_y$, to
\bneq
\begin{aligned}
\text{minimize} \quad &c_0(u_f(t), y_f(t)) + r(e_f(t)) \label{tpc} \\
\text{subject to} \quad &c_j(u_f(t), y_f(t)) \leq 0, \ j = 1, \dots, J \\
&y_f(t) = P z_p(t) + F u_f(t) + e_f(t) .
\end{aligned}
\eneq

The variables in \eqref{tpc} are the planned input and output trajectories $u_f(t) \in \R^{h n_u}$ and $y_f(t) \in \R^{h n_y}$, where $h$ is the prediction horizon, and $e_f(t) \in \R^{h n_y}$. Depending on the context, $e_f(t)$ can be viewed as a slack variable, as noise, or as a prediction or estimation error. We write the variables as
\[
u_f(t) = \bmat
u(1 | t) \\
\vdots \\
u(h | t) \\
\emat , \ y_f(t) = \bmat
y(1 | t) \\
\vdots \\
y(h | t) \\
\emat , \ e_f(t) = \bmat
e(1 | t) \\
\vdots \\
e(h | t) \\
\emat ,
\]
where the notation $i | t$ indicates a plan or prediction made $i$ steps ahead at time $t$. For example, $u(1 | t)$ is the plan for $u(t)$ made at time $t$. If the cost and constraint functions $c_0$, \dots, $c_J : \R^{h n_u} \times \R^{h n_y} \rightarrow \R$ and the regularizer $r : \R^{h n_y} \rightarrow \R \cup \setof{\infty}$ are convex, then \eqref{tpc} is a convex optimization problem.

\begin{algorithm}[t]
\caption{Trajectory predictive control}
\label{tpcAlg}
\begin{algorithmic}
\State {\bf Input:} 
Trajectory predictor matrices $P$ and $F$; cost and constraint functions $c_0$, \dots, $c_J$; regularization function $r$; initial input/output history $z_p(1)$

\For{$t = 1$, 2, \dots}
\bit
\item Solve \eqref{tpc} and implement $u(t) = u^\star(1 | t)$
\item Observe $y(t)$ and form
{\small
\[
z(t) =
\bmat
u(t) \\
y(t)
\emat, \ z_p(t + 1) =
\bmat
z(t - m + 1) \\
\vdots \\
z(t)
\emat
\]
}
\eit
\EndFor
\end{algorithmic}
\end{algorithm}

This paper focuses on the general trajectory predictor
\bneq
y_f(t) = P z_p(t) + F u_f(t) + e_f(t) . \label{predictor}
\eneq
We assume $P \in \R^{h n_y \times m n_z}$ and $F \in \R^{h n_y \times h n_u}$ are identified from a trajectory $\tilde z(1)$, \dots, $\tilde z(d)$ of input/output examples, with $e_f(t)$ treated as an estimation error to be minimized. We organize the training data (denoted by tildes) into Hankel matrices containing $n = d - m - h + 1$ trajectory examples:
\[
\begin{aligned}
\tilde Z_{\textcolor{black}{p}} &= \bmat \tilde z_p(m + 1) & \cdots & \tilde z_p(d - h + 1) \emat \in \R^{m n_z \times n} \\
\tilde U_{\textcolor{black}{f}} &= \bmat \tilde u_f(m + 1) & \cdots & \tilde u_f(d - h + 1) \emat \in \R^{h n_u \times n} \\
\tilde Y_{\textcolor{black}{f}} &= \bmat \tilde y_f(m + 1) & \cdots & \tilde y_f(d - h + 1) \emat \in \R^{h n_y \times n} , \\
\end{aligned}
\]
where
\[
\tilde u_{\textcolor{black}{f}}(t) = \bmat
\tilde u(t) \\
\vdots \\
\tilde u(t + h - 1) \\
\emat , \ \tilde y_{\textcolor{black}{f}}(t) = \bmat
\tilde y(t) \\
\vdots \\
\tilde y(t + h - 1) \\
\emat ,
\]
and $\tilde z_{\textcolor{black}{p}}(t)$ is constructed as in \eqref{z}.

Most of this paper will take the regularizer to be
\[
r(e_f(t)) = \delta_0(e_f(t)) = \begin{cases}
0  &\text{if } e_f(t) = 0 \\
\infty  &\text{otherwise.}
\end{cases}
\]
With $r = \delta_0$, the general TPC problem \eqref{tpc} is feasible only if $e_f(t) = 0$, so an equivalent problem is to
\bneq
\begin{aligned}
\text{minimize} \quad &c_0(u_f(t), y_f(t))  \\
\text{subject to} \quad &c_j(u_f(t), y_f(t)) \leq 0, \ j = 1, \dots, J \\
&y_f(t) = P z_p(t) + F u_f(t) , \label{delta0tpc}
\end{aligned}
\eneq
with variables $u_f(t),y_f(t)$. Less strict regularizers let the optimization steer the planned $u_f(t), y_f(t)$ away from the central prediction $u_f(t), P z_p(t) + F u_f(t)$ to reduce the cost $c_0(u_f(t), y_f(t))$, a tactic that several DDPC variants use.

DDPC research often restricts the cost to be quadratic and the constraints to be separable in the input and output. The formulation \eqref{tpc} does not make those restrictions in general but includes them as special cases.

\section{DDPC AND TRAJECTORY PREDICTORS}
\label{predictorSurvey}

The predictor structure \eqref{predictor} in the TPC problem \eqref{tpc} aligns with Willems' fundamental lemma of behavioral system theory \cite{willems2007behavioral, markovsky2021behavioral}, which applies if the underlying system is deterministic, LTI, and controllable, and if the training inputs $\tilde u(1)$, \dots, $\tilde u(d)$ are persistently exciting \cite{willems2005note, markovsky2023persistency}. Given the recent input/output history $z_p(t)$, Willems' fundamental lemma implies that $u_f(t), y_f(t)$ is a possible trajectory if and only if there exists an $\alpha(t) \in \R^n$ such that 
\bneq
\bmat
\tilde Z_{\textcolor{black}{p}} \\
\tilde U_{\textcolor{black}{f}} \\
\tilde Y_{\textcolor{black}{f}} \\
\emat \alpha(t) = \bmat
z_p(t) \\
u_f(t) \\
y_f(t) \\
\emat . \label{willems}
\eneq
For fixed $z_p(t)$ and $u_f(t)$, \eqref{willems} implies that $y_f(t) = \tilde Y_{\textcolor{black}{f}} \alpha(t)$ is a possible output trajectory for any $\alpha(t)$ satisfying
\bneq
\bmat
\tilde Z_{\textcolor{black}{p}} \\
\tilde U_{\textcolor{black}{f}} \\
\emat \alpha(t) = \bmat
z_p(t) \\
u_f(t) \\
\emat . \label{willemsPredictor}
\eneq
The dimension $n$ of $\alpha(t)$ may exceed the row rank of $\bmat \tilde Z_{\textcolor{black}{p}}^\top & \tilde U_{\textcolor{black}{f}}^\top \emat^\top$, which is $m n_z + h n_u$ under persistent excitation \cite{willems2005note, markovsky2023persistency}, so a solution $\alpha(t)$ to \eqref{willemsPredictor} may not be unique.

The DeePC method \cite{coulson2019data} replaces the trajectory predictor \eqref{predictor} in problem \eqref{delta0tpc} by \eqref{willems}, choosing $u_f(t)$, $y_f(t)$, $\alpha(t)$ to
\bneq
\begin{aligned}
\text{minimize} \quad &c_0(u_f(t), y_f(t)) \label{deepc} \\
\text{subject to} \quad &c_j(u_f(t), y_f(t)) \leq 0, \ j = 1, \dots, J \\
&\eqref{willems}.
\end{aligned}
\eneq
DeePC performs well for deterministic LTI systems, but under uncertainty or nonlinearity, the true system behavior can diverge wildly from the planned trajectory $u_f^\star(t), y_f^\star(t)$ corresponding to an optimal $\alpha^\star(t)$. In \cite{moffat2025bias}, Moffat et al. explain this in part through {\it optimism bias}: When many $\alpha(t)$ may satisfy \eqref{willems}, the DeePC optimization is free to choose an $\alpha^\star(t)$ that reduces the cost $c_0$. DeePC can also suffer from bias when training data are gathered in closed loop \cite{dinkla2023closed, moffat2025bias}.

\subsection{DeePC, SPC, \texorpdfstring{$\gamma$}{g}-DDPC, and the subspace predictor}

Connections between \eqref{willems} and \eqref{predictor} can be seen by finding the least-norm $\hat \alpha(t)$ satisfying \eqref{willemsPredictor}, which solves
\[
\begin{aligned}
\text{minimize} \quad &\norm{ \alpha(t) }_2  \\
\text{subject to} \quad &\bmat
\tilde Z_{\textcolor{black}{p}} \\
\tilde U_{\textcolor{black}{f}} \\
\emat \alpha(t) = \bmat
z_p(t) \\
u_f(t) \\
\emat .
\end{aligned}
\]
If $\bmat \tilde Z_{\textcolor{black}{p}}^\top & \tilde U_{\textcolor{black}{f}}^\top \emat^\top$ has full row rank, then 
\[
\begin{aligned}
\hat \alpha(t) &= \bmat
\tilde Z_{\textcolor{black}{p}} \\
\tilde U_{\textcolor{black}{f}} \\
\emat^\dagger \bmat
z_p(t) \\
u_f(t) \\
\emat .
\end{aligned}
\]
The output trajectory corresponding to $\hat \alpha(t)$ is
\[
\hat y_f(t) = \tilde Y_{\textcolor{black}{f}} \hat \alpha(t) = P_\text{sbs} z_p(t) + F_\text{sbs} u_f(t) 
\]
with
\bneq
\bmat P_\text{sbs} & F_\text{sbs} \emat = \tilde Y_{\textcolor{black}{f}} \bmat
\tilde Z_{\textcolor{black}{p}} \\
\tilde U_{\textcolor{black}{f}} \\
\emat^\dagger . \label{subspace}
\eneq
The predictor \eqref{predictor} with $P = P_\text{sbs}$ and $F = F_\text{sbs}$ from \eqref{subspace} is known as the {\it subspace predictor} \cite{moffat2025bias, favoreel1999spc, moffat2024transient}. TPC with the subspace predictor and $r = \delta_0$ is known as subspace predictive control (SPC) \cite{favoreel1999spc}. SPC is equivalent to DeePC augmented with the constraint $\alpha(t) = \hat \alpha(t)$, the least-norm solution to \eqref{willemsPredictor} \cite{sader2025causality}.

In \cite{breschi2023data}, Breschi et al. obtain an explicit formula for the subspace predictor from the LQ decomposition
\bneq
\bmat
\tilde Z_{\textcolor{black}{p}} \\
\tilde U_{\textcolor{black}{f}} \\
\tilde Y_{\textcolor{black}{f}} \\
\emat = \bmat
L_{11} & & \\
L_{21} & L_{22} & \\
L_{31} & L_{32} & L_{33} \\
\emat \bmat
Q_1 \\
Q_2 \\
Q_3 \\
\emat . \label{lq}
\eneq
Here $L_{11} \in \R^{m n_z \times m n_z}$, $L_{22} \in \R^{h n_u \times h n_u}$, and $L_{33} \in \R^{h n_y \times h n_y}$ are lower triangular; $L_{21} \in \R^{h n_u \times m n_z}$, $L_{31} \in \R^{h n_y \times m n_z}$, and $L_{32} \in \R^{h n_y \times h n_u}$ are generally dense; and $Q_1 \in \R^{m n_z \times n}$, $Q_2 \in \R^{h n_u \times n}$, and $Q_3 \in \R^{h n_y \times h n_y}$ are orthonormal: Each $Q_i Q_j^\top$ equals an identity matrix if $i = j$ and a zero matrix if $i \neq j$. If $\bmat \tilde Z_{\textcolor{black}{p}}^\top & \tilde U_{\textcolor{black}{f}}^\top \emat^\top$ has full row rank, then $L_{11}$ and $L_{22}$ are invertible and
\bneq
\begin{aligned}
\bmat P_\text{sbs} & F_\text{sbs} \emat &= \tilde Y_{\textcolor{black}{f}} \bmat Q_1^\top & Q_2^\top \emat \bmat L_{11} & \\ L_{21} & L_{22} \emat^{-1} \\
&= \tilde Y_{\textcolor{black}{f}} \bmat Q_1^\top & Q_2^\top \emat \bmat L_{11}^{-1} & \\ - L_{22}^{-1} L_{21} L_{11}^{-1} & L_{22}^{-1} \emat  \label{inverse} \\
\end{aligned}
\eneq
from proposition 3.9.7 of \cite{bernstein2009matrix}. But
$\tilde Y_{\textcolor{black}{f}} \bmat Q_1^\top & Q_2^\top \emat = \bmat L_{31} & L_{32} \emat$ by \eqref{lq} and orthonormality of the $Q_i$, so
\bneq
F_\text{sbs} = L_{32} L_{22}^{-1} , \ P_\text{sbs} = (L_{31} - F_\text{sbs} L_{21}) L_{11}^{-1} . \label{subspaceLQ}
\eneq

TPC with the subspace predictor in the form \eqref{subspaceLQ} is known as $\gamma$-DDPC \cite{breschi2023data}. In \cite{sader2025causality}, Sader et al. define regularized-$\gamma$-DDPC by adding a slack variable equivalent to $e_f(t)$ and regularizing it with $r(e_f(t)) = \lambda \norm{e_f(t)}_2^2$.

\subsection{Causal-\texorpdfstring{$\gamma$}{g}-DDPC and the multistep predictor}

In general, the subspace predictor is not causal: It models future inputs as influencing past outputs. To see this, we write $F$ in the general trajectory predictor \eqref{predictor} in block form as
\[
F = \bmat
F_{11} & \cdots & F_{1h} \\
\vdots & \ddots & \vdots \\
F_{h1} & \cdots & F_{hh} \\
\emat \text{ where } F_{ij} \in \R^{n_y \times n_u} .
\]

\begin{definition}
\label{causal}
The trajectory predictor \eqref{predictor} is {\bf causal} if $F$ is block lower triangular (BLT): $F_{ij} = 0$ for $i = 1$, \dots, $h - 1$ and $j = i + 1$, \dots, $h$. 
\end{definition}

Causality implies that the planned input $u(j | t)$ in \eqref{tpc} can only influence the planned output $y(i | t)$ if $i \geq j$. Although $L_{22}^{-1}$ is lower triangular, $L_{32}$ is dense in general. In the subspace predictor, $F_\text{sbs}$ is the product of the dense matrix $L_{32}$ and the lower triangular matrix $L_{22}^{-1}$, so in general $F_\text{sbs}$ is dense and the subspace predictor is not causal.

In \cite{sader2025causality}, Sader et al. investigate the {\it multistep predictor} $P_\text{mlt}$, $F_\text{mlt}$, defined as solving the constrained least squares problem
\[
\begin{aligned}
\text{minimize} \quad & \| \tilde Y_{\textcolor{black}{f}} - P \tilde Z_{\textcolor{black}{p}} - F \tilde U_{\textcolor{black}{f}} \|_\text{Fro}^2 \\
\text{subject to} \quad & F \text{ BLT} \\
\end{aligned}
\]
with variables $P$ and $F$. Sader et al. show that 
\bneq
F_\text{mlt} = \text{BLT}(L_{32}) L_{22}^{-1} , \ P_\text{mlt} = (L_{31} - F_\text{mlt} L_{21}) L_{11}^{-1} . \label{multistep}
\eneq
The operator BLT$(L_{32})$ zeros out the strictly block upper triangular part of $L_{32}$, returning a version of $L_{32}$ that is BLT but otherwise unchanged. From \eqref{multistep} and \eqref{subspaceLQ}, the multistep predictor is a `causalized' version of the subspace predictor.

In \cite{sader2025causality}, Sader et al. refer to TPC with the multistep predictor \eqref{multistep} and $r = \delta_0$ as causal-$\gamma$-DDPC. In numerical examples, causal-$\gamma$-DDPC outperforms $\gamma$-DDPC and SPC for small training datasets gathered in closed loop; adding the slack variable $e_f(t)$ with $r(e_f(t)) = \lambda \norm{e_f(t)}_2^2$, a variant that Sader et al. call regularized-causal-$\gamma$-DDPC, further improves closed-loop performance for small datasets.

\subsection{Transient predictive control and the transient predictor}

In \cite{moffat2024transient, moffat2025bias}, Moffat et al. develop another approach to estimating the multistep predictor. They define the {\it transient predictor} in terms of matrices $\Phi^p \in \R^{h n_y \times m n_z}$, $\Phi^u \in \R^{h n_y \times m n_u}$, $\Phi^y \in \R^{h n_y \times h n_y}$ satisfying
\bneq
\begin{aligned}
y_f(t) = \Phi^p z_p(t) + \Phi^u u_f(t) + \Phi^y y_f(t) + \varepsilon_f(t) , \label{rawTransient} \\
\end{aligned}
\eneq
where $\varepsilon_f(t) \in \R^{h n_y}$ is an error. To enforce causality, Moffat et al. require $\Phi^u$ to be BLT and $\Phi^y$ to be strictly BLT (meaning $\Phi^y$ is BLT and $\Phi_{11}^y = \dots = \Phi_{hh}^y = 0_{n_y, n_y}$). Rearranging \eqref{rawTransient} gives an instance of the trajectory predictor \eqref{predictor} with
\bneq
\begin{aligned}
P &= (I - \Phi^y)^{-1} \Phi^p, \ F = (I - \Phi^y)^{-1} \Phi^u \\ 
e_f(t) &= (I - \Phi^y)^{-1} \varepsilon_f(t) . \label{transient}
\end{aligned}
\eneq
The inverse $(I - \Phi^y)^{-1}$ exists because $I - \Phi^y$ is lower triangular with ones along its diagonal.

In \cite{moffat2024transient}, Moffat et al. form the transient predictor $P_\text{trn}$, $F_\text{trn}$ from \eqref{transient} with $\Phi^p$, $\Phi^u$, and $\Phi^y$ estimated from the LQ decomposition of the full Hankel matrix of input/output data. In \cite{moffat2025bias}, Moffat et al. refer to TPC with $r = \delta_0$ and the predictor \eqref{transient} as transient predictive control. They show that transient predictive control is not subject to optimism bias or to bias when training data are gathered in closed loop.

\subsection{Closed-loop SPC and the fixed-length predictor}

Dong et al. \cite{dong2008closed} and Chiuso et al. \cite{chiuso2025harnessing} form the {\it fixed-length predictor} as a structured special case of \eqref{rawTransient}: 
\[
\begin{aligned}
\Phi^p &= \bmat
\phi^p_1 & \phi^p_2 & \cdots &\phi^p_m \\
& \phi^p_1 & \cdots & \phi^p_{m-1} \\
& & \ddots & \vdots \\
\emat , \ \Phi^u = \bmat
\phi^u_1 & & \\
\vdots & \ddots & \\
\phi^u_h & \cdots & \phi^u_1 \\
\emat , \\
\end{aligned}
\]
with blocks $\phi^p_1$, \dots, $\phi_m^p \in \R^{n_y \times n_z}$, $\phi^u_1$, \dots, $\phi^u_h \in \R^{n_y \times n_u}$, and $\phi^y_1$, \dots, $\phi^y_{h-1} \in \R^{n_y \times n_y}$. The matrix $\Phi^y$ is structured similarly to $\Phi^u$ but strictly BLT. The (block) $h \times m$ matrix $\Phi^p$ may be wide, square, or tall. In any case, all elements below the $\phi^p_1$ block diagonal are zero.

In \cite{dong2008closed}, Dong et al. refer to TPC with the fixed-length predictor and $r = \delta_0$ as closed-loop SPC. They show that unlike the subspace predictor, the fixed-length predictor works with training data gathered in closed loop. In \cite{chiuso2025harnessing}, Chiuso et al. use the fixed-length predictor to provide conditions under which separating identification from the rest of DDPC design incurs no loss of optimality. In \cite{liu2025closed}, Liu and Jansson use a predictor with $\Phi^u$ and $\Phi^y$ structured as in the fixed-length predictor, but with $\Phi^p$ either unstructured (as in the transient predictor) or constrained to have low rank.

\section{MPC AND THE STATE-SPACE PREDICTOR}
\label{mainResult}

We define the {\it state-space predictor} such that the trajectory predictor \eqref{predictor} is equivalent to the LTI state-space model
\bneq
\begin{aligned}
z_p(t+1) &= A z_p(t) + B u(t) + K \varepsilon(t) \\
y(t) &= C z_p(t) + D u(t) + \varepsilon(t)  \label{ss} \\
\end{aligned}
\eneq
for appropriate choices of $A \in \R^{m n_z \times m n_z}$, $B \in \R^{m n_z \times n_u}$, $C \in \R^{h n_y \times m n_z}$, $D \in \R^{h n_y \times h n_u}$, $K \in \R^{m n_z \times n_y}$, and $\varepsilon(t) \in \R^{n_y}$. We begin with the following observation.

\begin{proposition}
\label{arx}
For any causal trajectory predictor of the form \eqref{predictor}, the one-step predictor (meaning the first $n_y$-block row of \eqref{predictor}) is a linear ARX model with autoregressive memory $m$, exogenous memory $m + 1$, and delay zero.
\end{proposition}

\begin{proof}
By definition \ref{causal}, $F_{12} = \dots = F_{1m} = 0$ for any causal trajectory predictor \eqref{predictor}, so the one-step-ahead predictor is
\bneq
y(t) = P_1 z_p(t) + F_{11} u(t) + e(t) . \label{y1causal}
\eneq
The first block row  of $P$, $P_1 \in \R^{n_y \times m n_z}$, can be written as
\[
P_1 = \bmat P_{11}^u & P_{11}^y & \cdots & P_{1m}^u & P_{1m}^y \emat  
\]
with $P_{1j}^u \in \R^{n_y \times n_u}$ and $P_{1j}^y \in \R^{n_y \times n_y}$. Recalling \eqref{z}, expanding $P_1$ and $z_p(t)$ in \eqref{y1causal} gives
\[
\begin{aligned}
&y(t) - P_{11}^y y(t-m) - \dots - P_{1m}^y y(t-1) \\
= \ &P_{11}^u u(t-m) + \dots + P_{1m}^u u(t-1) + F_{11} u(t) + e(t) .
\end{aligned}
\]
This is a linear ARX$(m,m+1,0)$ model.
\end{proof}

In light of Proposition \ref{arx}, we construct $A$, $B$, $C$, $D$, $K$, and $\varepsilon_f(t) = (\varepsilon(t)$, \dots, $\varepsilon(t + h - 1)) \in \R^{h n_y}$ such that the trajectory predictor \eqref{predictor} is consistent with iteratively applying the one-step-ahead ARX model. Lemma \ref{abcdk} provides suitable definitions of the system matrices.

\begin{lemma}
\label{abcdk}
Let $\varepsilon(t) = e(t)$, $C = P_1$, $D = F_{11}$,
\bneq
\begin{aligned}
A &= \bmat
\bmat 0_{(m-1) n_z, n_z} & I_{(m-1) n_z} \emat \\
0_{n_u, m n_z} \\
C \\
\emat \\ 
B &= \bmat
0_{(m-1) n_z, n_u} \\
I_{n_u} \\
D \\
\emat , \ K = \bmat
0_{(m-1) n_z + n_u, n_y} \\
I_{n_y} \\
\emat . \label{ab}
\end{aligned}
\eneq
Then $z_p(t)$, $u(t)$, and $y(t)$ satisfy the causal one-step predictor \eqref{y1causal} if and only if they satisfy the state-space model \eqref{ss}.
\end{lemma}

\begin{proof}
With the definitions of $C$, $D$, and $\varepsilon(t)$, the one-step predictor \eqref{y1causal} is equivalent to the output equation in \eqref{ss}. For the dynamics, the definition \eqref{z} of $z_p(t)$ implies that
\bneq
\begin{aligned}
&z_p(t+1) = \bmat
z(t - m + 1) \\
\vdots \\
z(t) \\
\emat = \bmat
0_{(m-1) n_z, n_z} \\
I_{n_z} \\
\emat z(t)  \label{app1} \\
&\quad + \bmat
0_{(m-1)n_z, n_z} & I_{(m-1)n_z} \\
0_{n_z} & 0_{n_z, (m-1) n_z} \\
\emat z_p(t) .
\end{aligned}
\eneq
Since $z(t) = (u(t), y(t))$ and $y(t) = C z_p(t) + D u(t) + \varepsilon(t)$, 
\[
\begin{aligned}
z(t) = \bmat 
I_{n_u} \\
0_{n_y,n_u} \\
\emat u(t) + \bmat
0_{n_u,n_y} \\
I_{n_y} \\
\emat ( C z_p(t) + D u(t) + \varepsilon(t) ) .
\end{aligned}
\]
Substituting this expression into \eqref{app1} gives $z_p(t+1) = A z_p(t) + B u(t) + K \varepsilon(t)$ with the $A$, $B$, and $K$ in \eqref{ab}.
\end{proof}

Lemma \ref{abcdk} establishes a correspondence between the state-space model \eqref{ss} and the one-step-ahead predictor \eqref{y1causal}. To build the full state-space predictor in the form \eqref{predictor}, we iteratively apply \eqref{ss}. Solving the output equation for $\varepsilon(t)$ and substituting it into the dynamics gives
\bneq
z_p(t+1) = \mathcal A z_p(t) + \mathcal B u(t) + K y(t) , \label{innovationsDynamics}
\eneq
where $\mathcal A = A - K C$ and $\mathcal B = B - K D$. For $i > 0$, iteratively applying this innovations-form model gives
\bneq
\begin{aligned}
&y(t + i) = C \mathcal A^i z_p(t) + D u(t + i) + \varepsilon(t + i)  \label{yti1}  \\
&+ C \sum_{j=0}^{i-1} \mathcal A^{i-j-1} ( \mathcal B u(t + j) + K y(t + j) ) . 
\end{aligned}
\eneq
Forming this model at each $i = 1$, \dots, $h - 1$ gives a predictor of the form \eqref{rawTransient} with 
\bneq
\begin{aligned}
\Phi^p &= \bmat
C \mathcal A^0 \\
\vdots \\
C \mathcal A^{h-1} \\
\emat , \ \Phi^u = \bmat
D & & & \\
C \mathcal A^0 \mathcal B & D & & \\
\vdots & \ddots & \ddots & \\
C \mathcal A^{h-2} \mathcal B & \cdots & C \mathcal A^0 \mathcal B & D \\
\emat \\ 
\Phi^y &= \bmat
 & & & \\
C \mathcal A^0 K & & & \\
\vdots & \ddots & & \\
C \mathcal A^{h-2} K & \cdots & C \mathcal A^0 K & \\
\emat . \label{ssPhi} 
\end{aligned}
\eneq
As expected from the state-space model \eqref{ss}, $\Phi^p$ is an observability matrix and the Toeplitz matrices $\Phi^u$ and $\Phi^y$ contain the Markov parameters $D$, $C \mathcal A^i \mathcal B$, and $C \mathcal A^i K$.

Given $\Phi^p$, $\Phi^u$, and $\Phi^y$, the state-space predictor matrices $P_\text{sts}$ and $F_\text{sts}$ can be formed according to \eqref{transient}. To quantify uncertainty, the one-step-ahead error covariance matrix estimate
\[
(\tilde Y_{\textcolor{black}{f1}} - P_1 \tilde Z_{\textcolor{black}{p}} - F_{11} \tilde U_{\textcolor{black}{f1}}) (\tilde Y_{\textcolor{black}{f1}} - P_1 \tilde Z_{\textcolor{black}{p}} - F_{11} \tilde U_{\textcolor{black}{f1}})^\top / n ,
\]
where $\tilde Y_{\textcolor{black}{f1}} \in \R^{n_y \times n}$ is the first $n_y$-block row of $\tilde Y_{\textcolor{black}{f}}$ and $\tilde U_{\textcolor{black}{f1}} \in \R^{n_u \times n}$ is the first $n_u$-block row $\tilde U_{\textcolor{black}{f}}$, or the trajectory error covariance matrix estimate
\[
(\tilde Y_{\textcolor{black}{f}} - P \tilde Z_{\textcolor{black}{p}} - F \tilde U_{\textcolor{black}{f}}) (\tilde Y_{\textcolor{black}{f}} - P \tilde Z_{\textcolor{black}{p}} - F \tilde U_{\textcolor{black}{f}})^\top / n
\]
can be computed, optionally subtracting the number of estimated parameters, $n_y (m n_z + n_u)$, in the denominators.

Alg. \ref{stsAlg} summarizes state-space predictor construction. Thm. \ref{stsTheorem} establishes its equivalence to the state-space model \eqref{ss}.

\begin{algorithm}[t]
\caption{State-space predictor identification}
\label{stsAlg}
\begin{algorithmic}
\State {\bf Input:} Memory $m$; prediction horizon $h$; training data matrices $\tilde Z_{\textcolor{black}{p}}$, $\tilde U_{\textcolor{black}{f1}}$, $\tilde Y_{\textcolor{black}{f1}}$

\bit
\item One-step-ahead fit: $\bmat C & D \emat = \tilde Y_{\textcolor{black}{f1}} \bmat
\tilde Z_{\textcolor{black}{p}} \\
\tilde U_{\textcolor{black}{f1}} \\
\emat^\dagger$

\item $A$, $B$, $K$ from \eqref{ab}, $\mathcal A = A - K C$, $\mathcal B = B - K D$ 

\item $\Phi^p$, $\Phi^u$, $\Phi^y$ from \eqref{ssPhi}

\item Return $P$, $F$ from \eqref{transient}
\eit
\end{algorithmic}
\end{algorithm}

\begin{theorem}
\label{stsTheorem}
Form $A$, $B$, $C$, $D$, $K$, $P$, $F$, and $\Phi^y$ according to Alg. \ref{stsAlg} and let $e_f(t) = (I - \Phi^y)^{-1} \varepsilon_f(t)$. Then $z_p(t)$, $u_f(t) = (u(t), \dots, u(t + h - 1))$, and $y_f(t) = (y(t), \dots, y(t + h - 1))$ satisfy \eqref{predictor} if and only if
\bneq
\begin{aligned}
z_p(t+i+1) &= A z_p(t+i) + B u(t+i) + K \varepsilon(t+i) \\ 
y(t+i) &= C z_p(t + i) + D u(t+i) + \varepsilon(t+i) \label{ssi}
\end{aligned}
\eneq
for $i = 0$, \dots, $h - 1$.
\end{theorem}

\begin{proof}
For the `if' direction, the discussion above showed that iteratively applying the state-space model \eqref{ss} with any $A$, $B$, $C$, $D$, and $K$ generates a predictor of the form \eqref{predictor} with $P$ and $F$ of the form \eqref{transient} and $\Phi^p$, $\Phi^u$, $\Phi^y$ of the form \eqref{ssPhi}.

The `only if' direction proceeds by induction. At $i = 0$, the first block row of $e_f(t) = (I - \Phi^y)^{-1} \varepsilon_f(t)$ gives $e(t) = \varepsilon(t)$, so the result holds by Lemma \ref{abcdk}. For the inductive step, we suppose that at some $i$, the result holds for $j = 0$, \dots, $i-1$. Under this assumption, iteratively applying \eqref{innovationsDynamics} gives
\[
z_p(t + i) = \mathcal A^i z_p(t) + \sum_{j=0}^{i-1} \mathcal A^{i-j-1} ( \mathcal B u(t + j) + K y(t + j) ) .
\]
As $P$, $F$, and $e_f(t)$ are constructed according to \eqref{transient} and \eqref{ssPhi}, $y(t + i)$ satisfies \eqref{yti1}. Comparing \eqref{yti1} to the expression above for $z_p(t + i)$ shows that the output equation in \eqref{ssi} holds. With the output equation in hand, the same argument from the proof of Lemma \ref{abcdk} shows that the dynamics equation in \eqref{ssi} holds with $A$, $B$, and $K$ defined as in \eqref{ab}.
\end{proof}

Thm. \ref{stsTheorem} implies that TPC with the state-space predictor is effectively a special case of MPC with the state-space model \eqref{ss}. More concretely, the TPC problem \eqref{delta0tpc} is equivalent to the following MPC problem:
\bneq
\begin{aligned}
&\text{minimize} \quad c_0( (u(1|t),\dots,u(h|t)) , (y(1|t),\dots,y(h|t)) ) \label{mpc} \\
&\text{subject to} \\
&c_j( (u(1|t),\dots,u(h|t)) , (y(1|t),\dots,y(h|t)) ) \leq 0, \\ 
&\quad j = 1, \dots, J \\
&x(1 | t) = z_p(t) , \ x(i + 1 | t) = A x(i | t) + B u(i | t) , \ i = 1, \dots, h \\
&y(i | t) = C x(i  | t) + D u(i | t) , \ i = 1, \dots, h , \\
\end{aligned}
\eneq
with variables $x(1|t)$, \dots, $x(h + 1 |t)$, $u(1|t)$, \dots, $u(h|t)$, $y(1|t)$, \dots, $y(h|t)$ and with $A$, $B$, $C$, and $D$ defined as in Lemma \ref{abcdk}. Corollary \ref{stspmpc} formalizes this observation.

\begin{corollary}
\label{stspmpc}
Define $A$, $B$, $C$, $D$, $P$, and $F$ as in Alg. \ref{stsAlg}. Then
\[
u_f^\star(t) = \bmat
u^\star(1|t) \\
\vdots \\
u^\star(h|t) \\
\emat , \ y_f^\star(t) = \bmat
y^\star(1|t) \\
\vdots \\
y^\star(h|t) \\
\emat
\]
are optimal for the TPC problem \eqref{delta0tpc} if and only if $u^\star(1|t)$, \dots, $u^\star(h|t)$, $y^\star(1|t)$, \dots, $y^\star(h|t)$, and $x^\star(1|t)$, \dots, $x^\star(h+1|t)$ are optimal for the MPC problem \eqref{mpc}.
\end{corollary}

\begin{proof}
The result follows from Thm. \ref{stsTheorem}, which establishes a bijection between the variables of \eqref{delta0tpc} and \eqref{mpc}.
\end{proof}

Unlike most output-feedback MPC methods, TPC with the state-space predictor does not require a state estimator. This is because the controller has perfect knowledge of the state $z_p(t)$, which is just the recent inputs $u(t-m)$, \dots, $u(t-1)$ interleaved with the recent outputs $y(t-m)$, \dots, $y(t-1)$.

\subsection{Trajectory predictor comparisons}

\begin{figure}
\centering
\includegraphics[width=0.8\columnwidth]{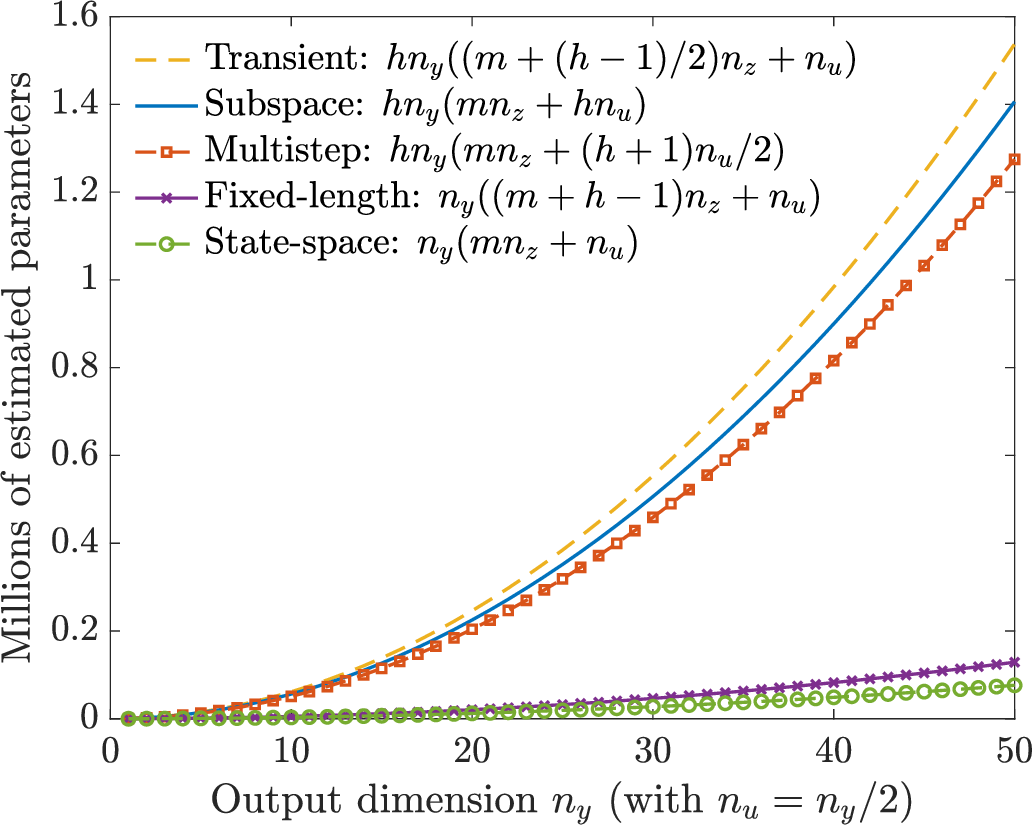}
\caption{Number of estimated parameters in each trajectory predictor. Formulas are general; curves use $m = 20$, $h = 15$, $n_u = n_y / 2$.
}
\label{paramFig}
\end{figure}

The subspace predictor is the only trajectory predictor discussed above that is not causal. Numerical experiments suggest that it may perform poorly when training data are gathered in closed loop \cite{dong2008closed, moffat2024transient, sader2025causality}. This may make the subspace predictor unsuitable for adaptive control or for applications where it is unsafe or impractical to excite the system with open-loop inputs to gather training data.

Different predictors require estimating different numbers of parameters. Fig. \ref{paramFig} shows how the number of estimated parameters in each predictor scales with the system dimensions. The formulas in the upper left are general; the curves are drawn for the example of $m = 20$, $h = 15$, and $n_u = n_y / 2$. The transient, subspace, and multistep predictors require estimating about an order of magnitude more parameters than the fixed-length and state-space predictors. The fixed-length and state-space predictors' favorable scaling may make them preferable for large-scale systems or for applications where it is costly or impractical to gather a large training dataset, as fitting models with more parameters generally requires more data. Because the state-space predictor only requires a one-step-ahead fit, it also makes slightly more efficient use of data, converting $\tilde z(1)$, \dots, $\tilde z(d)$ into $d - m$ identification examples, compared to $d - m - h + 1$ for the other predictors.

Fitting a trajectory predictor involves computing the pseudo-inverse of a data matrix. The data matrix must have full row rank for the pseudo-inverse to be well-defined in this context, so the data matrix must have at least as many columns as rows. This sets a minimum number of training examples required to fit the predictor. Table \ref{dataRequirements} shows the data requirements to fit each predictor. The state-space predictor requires the least data.

\begin{table}
\centering
\caption{Minimum number of training examples to fit each predictor}
\begin{tabular}{l | l | l}
 & Minimum \# of examples & \# less than transient \\
\hline
Transient & $(n_z+1)(m+h) - n_y - 1$ &  \\
Subspace & $(n_z + 1)m + (n_u + 1)h - 1$ & $(h - 1) n_y$ \\
Multistep & $(n_z + 1)m + (n_u + 1)h - 1$ & $(h - 1) n_y$ \\
Fixed-length & $(n_z + 1)m + h + n_u - 1$ & $(h - 1) n_z$ \\
State-space & $(n_z + 1) m + n_u$ & $(h - 1) (n_z + 1)$ \\
\end{tabular}
\label{dataRequirements}
\end{table}

\section{NUMERICAL EXPERIMENTS}
\label{experiments}

\begin{figure*}
\centering
\includegraphics[width=0.99\textwidth]{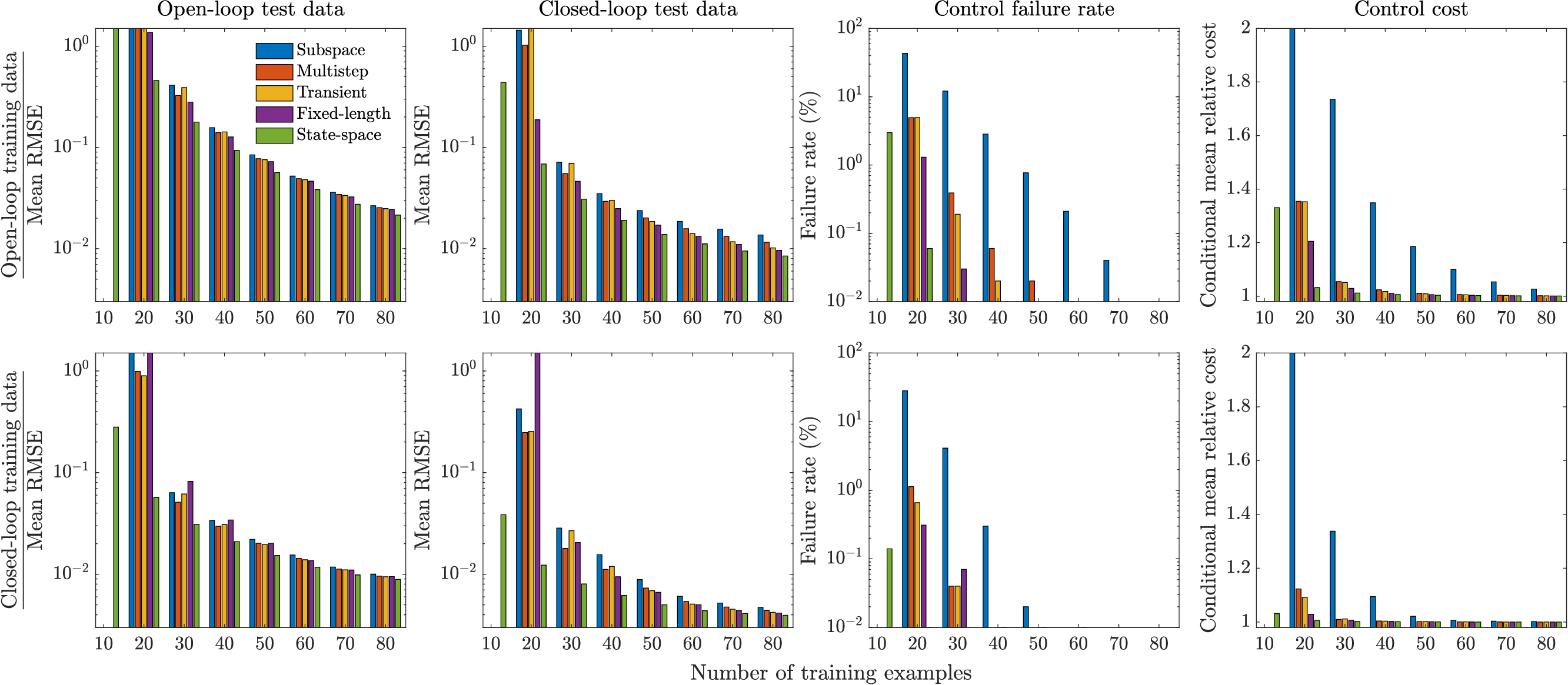}
\caption{Mean prediction RMSEs on test data gathered in open (first column) and closed loop (second) with training data gathered in open (top row) and closed loop (bottom) over 10,000 Monte Carlo runs. \textcolor{black}{Third column: Rates of failure, defined as incurring a cost more than 10 times the LQG cost in that Monte Carlo run. Fourth: Conditional mean cost relative to mean LQG cost, conditioned on the controller not failing.}}
\label{dFig}
\end{figure*}

To facilitate reproducibility and comparison with past results, we consider the double integrator example from \cite{moffat2025bias, moffat2024transient}. Moffat et al. Euler-discretize the continuous-time dynamics $\dot x_1 = x_2 + w_1$, $\dot x_2 = u + w_2$ with unit time step:
\[
\begin{aligned}
x(t+1) &= \bmat 1 & 1 \\ 0 & 1 \emat x(t) + 
\bmat 0 \\ 1 \emat 
u(t) + w(t) \\
y(t) &= x(t) + v(t) .
\end{aligned}
\]
All random variables are independent and (except the reference $y_r(t)$) zero-mean Gaussian. The disturbance $w(t)$ has covariance $\diag(0.0025, 0.0001)$. The noise $v(t)$ has covariance $0.0004 I_2$. To gather open-loop training data, we use 0.01-variance noise for $u(t)$. To gather closed-loop training data, we use discrete proportional-derivative control, $u(t) = - \bmat 0.0833 & 0.7944 \emat (y(t) - y_r(t))$ plus 0.01-variance noise for persistent excitation. The cost is 
\bneq
\sum_t (y(t) - y_r(t))^\top \diag(1000, 10) (y(t) - y_r(t)) + u(t)^2 \label{cost}
\eneq
with $y_r(t) = (y_{r1}(t), 0)$. We generate $y_{r1}$ as a sequence of steps with uniform random durations on $\setof{1, \dots, 50}$ and magnitudes on $[-5,5]$. There are no constraints beyond $y_f(t) = P z_p(t) + F u_f(t)$ (in the $r = \delta_0$ case), which allows elimination of $y_f(t)$. This reduces the TPC problem \eqref{delta0tpc} to an unconstrained convex quadratic program with solution 
\[
u^\star_f(t) = - (F^\top Q F + I_{h n_u})^{-1} F^\top Q (P z_p(t) - \hat y_{rf}(t) ) ,
\]
where $Q = I_h \otimes \diag(1000,10)$ and $\otimes$ denotes the Kronecker product. For the predicted reference trajectory $\hat y_{rf}(t) = (\hat y_r(t), \dots, \hat y_r(t + h - 1)) \in \R^{h n_y}$, we use the persistence model $\hat y_r(t+i) = y_r(t-1)$ for all $i = 0$, \dots, $h - 1$.

\subsection{Predictor generalization and closed-loop performance}
\label{sim1}

We evaluate TPC with the subspace, multistep, transient, fixed-length, and state-space predictors using Monte Carlo simulation. In each Monte Carlo run, we identify one version of each predictor from open-loop data and a second from closed-loop data. For each predictor and training dataset, we select the memory $m$ that minimizes the Akaike information criterion (AIC) \cite{akaike1974new}, averaged over elements of $y$ and over steps ahead in the prediction horizon. The AIC, an estimate of the information lost when representing the underlying data-generating process by a given model, generally suggests longer memories for larger training datasets. For the predictors and training dataset sizes investigated here, the AIC usually suggests a memory of one or two time steps and rarely more than three. We test prediction accuracy for both versions of each predictor on two fresh datasets, again gathering one test dataset in open loop and one in closed loop. For control, we test TPC with both versions of each predictor on the test datasets, benchmarked against an oracle LQG controller that has perfect knowledge of the true system model. \textcolor{black}{In all Monte Carlo runs and all data configurations, the state-space predictor system was controllable, as is the true system.} We use a prediction horizon of $h = 10$ time steps for all predictors in all data configurations. With this horizon, TPC with all predictors performs well given sufficient training data.

The first and second columns in Fig. \ref{dFig} show how prediction errors scale with the training dataset size $d$ for each predictor. Bar heights are sample means over 10,000 Monte Carlo runs with a test duration of 400 time steps. Training datasets are gathered in open loop for the top row of plots and closed loop for the bottom. The first and second columns show prediction root mean squared errors (RMSEs) on test data gathered in open and closed loop, respectively. RMSEs are averaged over elements of $y$ and over steps-ahead predictors. \textcolor{black}{The state-space predictor is formable with the smallest training dataset (10 examples), when the other predictors' data matrices are rank-deficient.} Test RMSEs decrease with the training dataset size for all predictors. The state-space predictor has the lowest RMSE of all predictors in all test cases, although RMSE differences decrease as the training dataset size increases. All predictors have lower RMSEs predicting closed-loop test data, even if they are trained on open-loop data. Similarly, all predictors have lower RMSEs when trained on closed-loop data, even if they are predicting open-loop test data.

\textcolor{black}{The third and fourth columns in Fig. \ref{dFig} show how TPC performance scales with the training dataset size for each predictor. The third column shows failure rates, with failure defined as incurring a cost \eqref{cost} more than 10 times the LQG cost in that Monte Carlo run. The fourth column shows the conditional mean cost, normalized by the mean LQG cost and conditioned on the event that the controller does not fail. For all predictors, both failure rates and conditional mean costs are higher with training data gathered in open loop than in closed loop and decrease as the training dataset size increases. The state-space predictor has the lowest failure rate and the lowest conditional mean cost of all predictors in all test cases. For all predictors, the failure rate converges to zero and the conditional mean cost converges to the mean LQG cost. Convergence is fastest with the state-space predictor and slowest with the subspace predictor.}

Fig. \ref{dFig} shows that the least-structured predictor (the subspace predictor) generally has the worst predictive accuracy and closed-loop performance. Prediction accuracy and closed-loop performance are generally better for predictors with more internal structure. The largest improvements come from imposing causal structure on the subspace predictor, which produces the multistep predictor, and from imposing state-space structure. \textcolor{black}{The most structured predictor (the state-space predictor) has the fewest parameters and the best predictive accuracy and closed-loop performance. Differences in predictive accuracy and closed-loop performance across predictors are largest for small training datasets and tend to vanish as the training dataset grows.} Based on these observations, we interpret structural constraints on the predictor as a form of implicit regularization that improves generalization by mitigating the risk of overfitting small training datasets.

\subsection{Relaxation and regularization}

The examples in \S\ref{sim1} used the regularizer $r = \delta_0$, reducing the general TPC problem \eqref{tpc} to \eqref{delta0tpc}. This section explores relaxing the equality constraint in \eqref{delta0tpc} by introducing the slack variable $e_f(t)$ and augmenting the cost function with Tikhonov regularization, giving \eqref{tpc} with $r(e_f(t)) = \lambda \norm{e_f(t)}_2^2$. Eliminating $y_f(t)$ reduces \eqref{tpc} to an unconstrained convex quadratic program in $u_f(t), e_f(t)$ with solution satisfying
\[
\bmat
F^\top Q F & F^\top Q \\
Q F & Q + \lambda I \\
\emat \bmat
u_f^\star(t) \\
e_f^\star(t) \\
\emat = \bmat 
F^\top \\
I \\
\emat Q ( P z_p(t) - \hat y_{rf}(t) ) .
\]
The matrix on the lefthand side is invertible if $F$ has full column rank. We implement the `relax-and-regularize' approach to TPC with the state-space predictor identified from $d = 50$ closed-loop training examples and regularization weight $\lambda = 0.1$. The simulation setup is otherwise the same as in \S\ref{sim1}.

TPC with the relax-and-regularize approach uses less control effort than TPC with $r = \delta_0$, at the cost of worse tracking performance. We interpret this as an example of optimism bias \cite{moffat2025bias}: Relaxing the equality constraint allows the optimization to choose a favorable realization of $e_f(t)$ that would help steer the output toward the reference. In this sense, the relax-and-regularize approach can be viewed as conceptually opposite to the robust optimization approach, which would view $e_f(t)$ as chosen adversarially and hedge against the worst-case realization. On average over 10,000 Monte Carlo runs, relaxing the equality constraint and regularizing $e_f(t)$ increases the mean cost realized in closed loop by 14\%.

\section{FUTURE WORK}
\label{futureWork}

There are several opportunities to extend this work. Certificates of stability and recursive feasibility could be formalized by reducing TPC to MPC via the state-space predictor. Robust, stochastic, scenario, or adaptive TPC implementations could be developed similarly. The relax-and-regularize approach to TPC could be explored more deeply, possibly drawing connections to optimistic exploration in reinforcement learning. TPC could be evaluated for time-varying or nonlinear systems and modified if necessary. TPC could be demonstrated in hardware and its deployment effort and closed-loop performance compared to existing methods.

\section{ACKNOWLEDGMENTS}

LDRP gratefully acknowledges support from the NSF's Graduate Research Fellowship Program. LDRP and AJK thank ASHRAE for support from the Grant-In-Aid Award.

\section*{REFERENCES}
\vspace{-1.5\baselineskip}
\bibliography{IEEEabrv,tpc}

\end{document}